\documentclass[12pt,oneside]{amsart}

\usepackage{amssymb}
\usepackage{amsfonts}
\usepackage{amscd}
\usepackage[matrix, arrow, curve]{xy}
\usepackage{graphicx}

\newtheorem{theorem}{Theorem}

\newtheorem{proposition}[theorem]{Proposition}

\theoremstyle{definition}

\theoremstyle{remark}

\newtheorem{example}{Example}

\begin{document}
\newcommand{\M}{\mathcal{M}}
\newcommand{\F}{\mathcal{F}}

\newcommand{\Teich}{\mathcal{T}_{g,N+1}^{(1)}}
\newcommand{\T}{\mathrm{T}}
\newcommand{\corr}{\bf}
\newcommand{\vac}{|0\rangle}
\newcommand{\Ga}{\Gamma}
\newcommand{\new}{\bf}
\newcommand{\define}{\def}
\newcommand{\redefine}{\def}
\newcommand{\Cal}[1]{\mathcal{#1}}
\renewcommand{\frak}[1]{\mathfrak{{#1}}}
\newcommand{\Hom}{\rm{Hom}\,}
\newcommand{\refE}[1]{(\ref{E:#1})}
\newcommand{\refCh}[1]{Chapter~\ref{Ch:#1}}
\newcommand{\refS}[1]{Section~\ref{S:#1}}
\newcommand{\refSS}[1]{Section~\ref{SS:#1}}
\newcommand{\refT}[1]{Theorem~\ref{T:#1}}
\newcommand{\refO}[1]{Observation~\ref{O:#1}}
\newcommand{\refP}[1]{Proposition~\ref{P:#1}}
\newcommand{\refD}[1]{Definition~\ref{D:#1}}
\newcommand{\refC}[1]{Corollary~\ref{C:#1}}
\newcommand{\refL}[1]{Lemma~\ref{L:#1}}
\newcommand{\R}{\ensuremath{\mathbb{R}}}
\newcommand{\C}{\ensuremath{\mathbb{C}}}
\newcommand{\N}{\ensuremath{\mathbb{N}}}
\newcommand{\Q}{\ensuremath{\mathbb{Q}}}
\renewcommand{\P}{\ensuremath{\mathcal{P}}}
\newcommand{\Z}{\ensuremath{\mathbb{Z}}}
\newcommand{\kv}{{k^{\vee}}}
\renewcommand{\l}{\lambda}
\newcommand{\gb}{\overline{\mathfrak{g}}}
\newcommand{\dt}{\tilde d}     
\newcommand{\hb}{\overline{\mathfrak{h}}}
\newcommand{\g}{\mathfrak{g}}
\newcommand{\h}{\mathfrak{h}}
\newcommand{\gh}{\widehat{\mathfrak{g}}}
\newcommand{\ghN}{\widehat{\mathfrak{g}_{(N)}}}
\newcommand{\gbN}{\overline{\mathfrak{g}_{(N)}}}
\newcommand{\tr}{\mathrm{tr}}
\newcommand{\gln}{\mathfrak{gl}(n)}
\newcommand{\son}{\mathfrak{so}(n)}
\newcommand{\spnn}{\mathfrak{sp}(2n)}
\newcommand{\sln}{\mathfrak{sl}}
\newcommand{\sn}{\mathfrak{s}}
\newcommand{\so}{\mathfrak{so}}
\newcommand{\spn}{\mathfrak{sp}}
\newcommand{\tsp}{\mathfrak{tsp}(2n)}
\newcommand{\gl}{\mathfrak{gl}}
\newcommand{\slnb}{{\overline{\mathfrak{sl}}}}
\newcommand{\snb}{{\overline{\mathfrak{s}}}}
\newcommand{\sob}{{\overline{\mathfrak{so}}}}
\newcommand{\spnb}{{\overline{\mathfrak{sp}}}}
\newcommand{\glb}{{\overline{\mathfrak{gl}}}}
\newcommand{\Hwft}{\mathcal{H}_{F,\tau}}
\newcommand{\Hwftm}{\mathcal{H}_{F,\tau}^{(m)}}

\newcommand{\car}{{\mathfrak{h}}}    
\newcommand{\bor}{{\mathfrak{b}}}    
\newcommand{\nil}{{\mathfrak{n}}}    
\newcommand{\vp}{{\varphi}}
\newcommand{\bh}{\widehat{\mathfrak{b}}}  
\newcommand{\bb}{\overline{\mathfrak{b}}}  
\newcommand{\Vh}{\widehat{\mathcal V}}
\newcommand{\KZ}{Kniz\-hnik-Zamo\-lod\-chi\-kov}
\newcommand{\TUY}{Tsuchia, Ueno  and Yamada}
\newcommand{\KN} {Kri\-che\-ver-Novi\-kov}
\newcommand{\pN}{\ensuremath{(P_1,P_2,\ldots,P_N)}}
\newcommand{\xN}{\ensuremath{(\xi_1,\xi_2,\ldots,\xi_N)}}
\newcommand{\lN}{\ensuremath{(\lambda_1,\lambda_2,\ldots,\lambda_N)}}
\newcommand{\iN}{\ensuremath{1,\ldots, N}}
\newcommand{\iNf}{\ensuremath{1,\ldots, N,\infty}}

\newcommand{\tb}{\tilde \beta}
\newcommand{\tk}{\tilde \varkappa}
\newcommand{\ka}{\kappa}
\renewcommand{\k}{\varkappa}
\newcommand{\ce}{{c}}

\newcommand{\Pif} {P_{\infty}}
\newcommand{\Pinf} {P_{\infty}}
\newcommand{\PN}{\ensuremath{\{P_1,P_2,\ldots,P_N\}}}
\newcommand{\PNi}{\ensuremath{\{P_1,P_2,\ldots,P_N,P_\infty\}}}
\newcommand{\Fln}[1][n]{F_{#1}^\lambda}
\newcommand{\tang}{\mathrm{T}}
\newcommand{\Kl}[1][\lambda]{\can^{#1}}
\newcommand{\A}{\mathcal{A}}
\newcommand{\U}{\mathcal{U}}
\newcommand{\V}{\mathcal{V}}
\newcommand{\W}{\mathcal{W}}
\renewcommand{\O}{\mathcal{O}}
\newcommand{\Ae}{\widehat{\mathcal{A}}}
\newcommand{\Ah}{\widehat{\mathcal{A}}}
\newcommand{\La}{\mathcal{L}}
\newcommand{\Le}{\widehat{\mathcal{L}}}
\newcommand{\Lh}{\widehat{\mathcal{L}}}
\newcommand{\eh}{\widehat{e}}
\newcommand{\Da}{\mathcal{D}}
\newcommand{\kndual}[2]{\langle #1,#2\rangle}
\newcommand{\cins}{\frac 1{2\pi\mathrm{i}}\int_{C_S}}
\newcommand{\cinsl}{\frac 1{24\pi\mathrm{i}}\int_{C_S}}
\newcommand{\cinc}[1]{\frac 1{2\pi\mathrm{i}}\int_{#1}}
\newcommand{\cintl}[1]{\frac 1{24\pi\mathrm{i}}\int_{#1 }}
\newcommand{\w}{\omega}
\newcommand{\ord}{\operatorname{ord}}
\newcommand{\res}{\operatorname{res}}
\newcommand{\nord}[1]{:\mkern-5mu{#1}\mkern-5mu:}
\newcommand{\codim}{\operatorname{codim}}
\newcommand{\ad}{\operatorname{ad}}
\newcommand{\Ad}{\operatorname{Ad}}
\newcommand{\supp}{\operatorname{support}}

\newcommand{\Fn}[1][\lambda]{\mathcal{F}^{#1}}
\newcommand{\Fl}[1][\lambda]{\mathcal{F}^{#1}}
\renewcommand{\Re}{\mathrm{Re}}

\newcommand{\ha}{H^\alpha}

\define\ldot{\hskip 1pt.\hskip 1pt}
\define\ifft{\qquad\text{if and only if}\qquad}
\define\a{\alpha}
\redefine\d{\delta}
\define\w{\omega}
\define\ep{\epsilon}
\redefine\b{\beta} \redefine\t{\tau} \redefine\i{{\,\mathrm{i}}\,}
\define\ga{\gamma}
\define\cint #1{\frac 1{2\pi\i}\int_{C_{#1}}}
\define\cintta{\frac 1{2\pi\i}\int_{C_{\tau}}}
\define\cintt{\frac 1{2\pi\i}\oint_{C}}
\define\cinttp{\frac 1{2\pi\i}\int_{C_{\tau'}}}
\define\cinto{\frac 1{2\pi\i}\int_{C_{0}}}
\define\cinttt{\frac 1{24\pi\i}\int_C}
\define\cintd{\frac 1{(2\pi \i)^2}\iint\limits_{C_{\tau}\,C_{\tau'}}}
\define\dintd{\frac 1{(2\pi \i)^2}\iint\limits_{C\,C'}}
\define\cintdr{\frac 1{(2\pi \i)^3}\int_{C_{\tau}}\int_{C_{\tau'}}
\int_{C_{\tau''}}}
\define\im{\operatorname{Im}}
\define\re{\operatorname{Re}}
\define\res{\operatorname{res}}
\redefine\deg{\operatornamewithlimits{deg}}
\define\ord{\operatorname{ord}}
\define\rank{\operatorname{rank}}
\define\fpz{\frac {d }{dz}}
\define\dzl{\,{dz}^\l}
\define\pfz#1{\frac {d#1}{dz}}

\define\K{\Cal K}
\define\U{\Cal U}
\redefine\O{\Cal O}
\define\He{\text{\rm H}^1}
\redefine\H{{\mathrm{H}}}
\define\Ho{\text{\rm H}^0}
\define\A{\Cal A}
\define\Do{\Cal D^{1}}
\define\Dh{\widehat{\mathcal{D}}^{1}}
\redefine\L{\Cal L}
\newcommand{\ND}{\ensuremath{\mathcal{N}^D}}
\redefine\D{\Cal D^{1}}
\define\KN {Kri\-che\-ver-Novi\-kov}
\define\Pif {{P_{\infty}}}
\define\Uif {{U_{\infty}}}
\define\Uifs {{U_{\infty}^*}}
\define\KM {Kac-Moody}
\define\Fln{\Cal F^\lambda_n}
\define\gb{\overline{\mathfrak{ g}}}
\define\G{\overline{\mathfrak{ g}}}
\define\Gb{\overline{\mathfrak{ g}}}
\redefine\g{\mathfrak{ g}}
\define\Gh{\widehat{\mathfrak{ g}}}
\define\gh{\widehat{\mathfrak{ g}}}
\define\Ah{\widehat{\Cal A}}
\define\Lh{\widehat{\Cal L}}
\define\Ugh{\Cal U(\Gh)}
\define\Xh{\hat X}
\define\Tld{...}
\define\iN{i=1,\ldots,N}
\define\iNi{i=1,\ldots,N,\infty}
\define\pN{p=1,\ldots,N}
\define\pNi{p=1,\ldots,N,\infty}
\define\de{\delta}

\define\kndual#1#2{\langle #1,#2\rangle}
\define \nord #1{:\mkern-5mu{#1}\mkern-5mu:}
\newcommand{\MgN}{\mathcal{M}_{g,N}} 
\newcommand{\MgNeki}{\mathcal{M}_{g,N+1}^{(k,\infty)}} 
\newcommand{\MgNeei}{\mathcal{M}_{g,N+1}^{(1,\infty)}} 
\newcommand{\MgNekp}{\mathcal{M}_{g,N+1}^{(k,p)}} 
\newcommand{\MgNkp}{\mathcal{M}_{g,N}^{(k,p)}} 
\newcommand{\MgNk}{\mathcal{M}_{g,N}^{(k)}} 
\newcommand{\MgNekpp}{\mathcal{M}_{g,N+1}^{(k,p')}} 
\newcommand{\MgNekkpp}{\mathcal{M}_{g,N+1}^{(k',p')}} 
\newcommand{\MgNezp}{\mathcal{M}_{g,N+1}^{(0,p)}} 
\newcommand{\MgNeep}{\mathcal{M}_{g,N+1}^{(1,p)}} 
\newcommand{\MgNeee}{\mathcal{M}_{g,N+1}^{(1,1)}} 
\newcommand{\MgNeez}{\mathcal{M}_{g,N+1}^{(1,0)}} 
\newcommand{\MgNezz}{\mathcal{M}_{g,N+1}^{(0,0)}} 
\newcommand{\MgNi}{\mathcal{M}_{g,N}^{\infty}} 
\newcommand{\MgNe}{\mathcal{M}_{g,N+1}} 
\newcommand{\MgNep}{\mathcal{M}_{g,N+1}^{(1)}} 
\newcommand{\MgNp}{\mathcal{M}_{g,N}^{(1)}} 
\newcommand{\Mgep}{\mathcal{M}_{g,1}^{(p)}} 
\newcommand{\MegN}{\mathcal{M}_{g,N+1}^{(1)}} 

\define \sinf{{\widehat{\sigma}}_\infty}
\define\Wt{\widetilde{W}}
\define\St{\widetilde{S}}
\newcommand{\SigmaT}{\widetilde{\Sigma}}
\newcommand{\hT}{\widetilde{\frak h}}
\define\Wn{W^{(1)}}
\define\Wtn{\widetilde{W}^{(1)}}
\define\btn{\tilde b^{(1)}}
\define\bt{\tilde b}
\define\bn{b^{(1)}}
\define \ainf{{\frak a}_\infty} 

%
\define\eps{\varepsilon}    
\newcommand{\e}{\varepsilon}
\define\doint{({\frac 1{2\pi\i}})^2\oint\limits _{C_0}
       \oint\limits _{C_0}}                            
\define\noint{ {\frac 1{2\pi\i}} \oint}   
\define \fh{{\frak h}}     
\define \fg{{\frak g}}     
\define \GKN{{\Cal G}}   
\define \gaff{{\hat\frak g}}   
\define\V{\Cal V}
\define \ms{{\Cal M}_{g,N}} 
\define \mse{{\Cal M}_{g,N+1}} 
\define \tOmega{\Tilde\Omega}
\define \tw{\Tilde\omega}
\define \hw{\hat\omega}
\define \s{\sigma}
\define \car{{\frak h}}    
\define \bor{{\frak b}}    
\define \nil{{\frak n}}    
\define \vp{{\varphi}}
\define\bh{\widehat{\frak b}}  
\define\bb{\overline{\frak b}}  
\define\KZ{Knizhnik-Zamolodchikov}
\define\ai{{\alpha(i)}}
\define\ak{{\alpha(k)}}
\define\aj{{\alpha(j)}}
\newcommand{\calF}{{\mathcal F}}
\newcommand{\ferm}{{\mathcal F}^{\infty /2}}
\newcommand{\Aut}{\operatorname{Aut}}
\newcommand{\End}{\operatorname{End}}
\newcommand{\laxgl}{\overline{\mathfrak{gl}}}
\newcommand{\laxsl}{\overline{\mathfrak{sl}}}
\newcommand{\laxso}{\overline{\mathfrak{so}}}
\newcommand{\laxsp}{\overline{\mathfrak{sp}}}
\newcommand{\laxs}{\overline{\mathfrak{s}}}
\newcommand{\laxg}{\overline{\frak g}}
\newcommand{\bgl}{\laxgl(n)}
\newcommand{\tX}{\widetilde{X}}
\newcommand{\tY}{\widetilde{Y}}
\newcommand{\tZ}{\widetilde{Z}}


\title[Integrable systems of algebraic origin]{Some integrable systems of algebraic origin and separation of variables}
\author[O.K. Sheinman]{O.K. Sheinman}
\address{Department of Geometry and Topology, Steklov Mathematical Institute, Moscow}

\maketitle



The coefficients of an algebraic plane curve passing through a given set of points Poisson commute with respect to each pair of coordinates corresponding to the given points. This fact has been observed in \cite{BT1,BT2} in the context of separation of variables for non-linear integrable systems. It is a particular case of this statement that the coefficients of the Lagrange interpolation polynomial commute with respect to Poisson brackets of the similar type on interpolation data. This fact has been given with an independent proof in \cite{Sh_arx2017} where the corresponding integrable system had emerged as a reduction of a rank 2 hyperelliptic Hitchin system (see also \cite{Sh_DAN2018}).

The aim of the present note is to draw attention to the above, and other integrable systems of pure algebraic origin, not necessary related to curves, and derive them in the general set-up of separation of variables. The method we use shows that separation of variables is applicable for generating the integrable systems while conventionally it starts from given ones.

We consider a system of equations
\begin{equation}\label{E:non_sys}
F_i(H_1,\ldots,H_n,a_i,b_i)=0,\quad i=1,\ldots,n
\end{equation}
where $F_i$ are given smooth functions, $H_j=H_j(a_1,\ldots,a_n,b_1,\ldots,b_n)$ are unknown smooth functions ($i,j=1,\ldots,n$) in complex (resp., real) variables. It is important that for any $i=1,\ldots,n$ the function $F_i$ explicitly depends on only one pair of variables $a_i,b_i$ (however it may depend on the remainder of variables via $H_1,\ldots,H_n$).

An important particular case is given by linear systems. Let $f(x,y)=\left( f_{ij}(x,y) \right)_{i=0,1,\ldots,n}^{j=0,1,\ldots,n}$ be a differentiable matrix-valued function of two real or complex variables. Consider the following system of linear equations with unknowns $H_1,\ldots,H_n$:
\begin{equation}\label{E:bas_sys}
   \sum_{j=1}^n f_{ij}(a_i,b_i)H_j=f_{i0}(a_i,b_i),\quad i=1,\ldots,n
\end{equation}
defined in the domain in $\C^{2n}$ where the matrix of the system with eliminated $k$-th equation has rank $n-1$ for all or some of $k:\, 1\le k\le n$. We think of the following basic examples with relation to the system~\refE{bas_sys}.
\begin{example}
Given $n$ pairs $(a_i,b_i)$ of real (complex) numbers ($i=1,\ldots,n$), the problem of finding the degree $n-1$ polynomial taking the value $b_i$ at the point $a_i$ for all $i=1,\ldots,n$ descends to the system of type \refE{bas_sys} with $f_{ij}(x,y)=x^{j-1}$ for $j\ne 0$, $f_{i0}(x,y)=y$ for all $i=1,\ldots,n$ (it can be also directly written in the form of Lagrange interpolation polynomial). It follows from \refP{comm} below that we may don't restrict ourselves with just a polynomial of degree $n-1$ given by $n$ pairs of numbers. We may consider polynomials of an arbitrary degree with only $n$ nonzero coefficients, and  Laurent polynomials.
\end{example}
\begin{example}
Consider a plane algebraic curve and write its equation in the form $\sum\limits_{j=1}^dH_jx^{p_j}y^{q_j}=f_0(x,y)$. We regard the coefficients on the right hand side of the equation as given, and the coefficients on the left hand side as subjected to finding from the requirement that the curve passes through $d$ points $(a_1,b_1),\ldots,(a_d,b_d)$. This leads to the system of equations of type \refE{bas_sys} with $f_{ij}(x,y)=x^{p_j}y^{q_j}$, $f_{i0}(x,y)=f_0(x,y)$ (with $f_{ij}$ independent of $i$). For example take  the Weierstra\ss\ model of an affine curve \cite{Buchstaber}. It is given by a pair of coprime natural numbers $(n,s)$ by means of the equation
\[
 y^n=x^s+\sum_{q>0}\l_q\sum_{\substack{ns-ni-sj=q\\ i,j\ge 0}}x^iy^j.
\]
Let $d$ be the number of different values of $q$ such that $q=ns-ni-sj>0$, $i,j\ge 0$. Given $d$ points $(a_k,b_k)\in\C^2$, $k=1,\ldots,d$ define the coefficients $\l_q$ from the requirement that our curve passes through all the points $(a_k,b_k)$. It descends to the system of the type \refE{bas_sys} with $n=d$,
\[
f_{k,q}(x,y)=\sum_{\substack{ns-ni-sj=q\\ i,j\ge 0}}x^iy^j \ (q>0), \ f_{k,0}(x,y)=y^n-x^s,\ k=1,\ldots,d.
\]
\end{example}
\begin{example}
Consider a polynomial given by its values at certain points, and by the values of its derivatives of arbitrary orders at the other, mutually different, points, so that there is a balance between unknown coefficients and given pairs of the type \emph{point--value}. It is a kind of the Hermit interpolation polynomial, in the classical one the values of the polynomial itself and its derivatives are given at the same points. It is easy to see that in this case we obtain a system of the type \refE{bas_sys} where $f_{ij}$ do depend on $i$ (in contrary to previous examples).
\end{example}
We consider a Poisson bracket of the form
\[
\{ f,g\}=\sum\limits_{j=1}^np_j\left(\frac{\partial f}{\partial a_j}\frac{\partial g}{\partial b_j}-\frac{\partial g}{\partial a_j}\frac{\partial f}{\partial b_j}\right)
\]
where $p_j=p_j(a_j,b_j)$ are smooth functions in only one pair of variables (to satisfy the Jacobi identity). For $p_j\equiv 1$ for all $j=1,\ldots,n$ we obtain the \emph{canonical Poisson bracket}. In the case $p_j\equiv 0$, $j\ne k$ we obtain a bracket corresponding to the pair $(a_k,b_k)$ of coordinates  (in \cite{BT1,BT2} $p_j$ is the same for all $j$ though it is of no importance for the proof given there).
\begin{proposition}\label{P:rels1}
The components $H_1,\ldots,H_n$ of the solution to \refE{bas_sys} considered as functions in $a_1,\ldots,a_n$, $b_1,\ldots,b_n$ defined in the domain in $\C^{2n}$ where $\rank\frac{\partial(F_1,\ldots,F_{k-1},F_{k+1}\ldots,F_n)} {\partial(H_1,\ldots,H_{k-1},H_k,H_{k+1}\ldots,H_n)}=n-1$
for a $k:\, 1\le k\le n$ satisfy the following relations:
\newline
$1^\circ$
if the vector $(\frac{\partial H_1}{\partial b_k},\ldots,\frac{\partial H_n}{\partial b_k})$ is not identically equal to $0$ then
\begin{equation}\label{E:rels2}
  \frac{\partial H_j}{\partial a_k}=M_k\frac{\partial H_j}{\partial b_k},\
  j=1,\ldots,n;
\end{equation}
\newline
$2^\circ$ for a $j_k$ such that $\det \frac{\partial(F_1,\ldots,F_{k-1},F_{k+1}\ldots,F_n)} {\partial(H_1,\ldots,H_{j_k-1},H_{j_k+1}\ldots,H_n)}\ne 0$
\begin{equation}\label{E:rels2}
  \frac{\partial H_j}{\partial a_k}=A_{jkj_k}\frac{\partial H_{j_k}}{\partial a_k},\
  \frac{\partial H_j}{\partial b_k}=A_{jkj_k}\frac{\partial H_{j_k}}{\partial b_k},\
  j=1,\ldots,n
\end{equation}
where $M_k$ and $A_{jkj_k}$ are functions in $a_1,\ldots,a_n$, $b_1,\ldots,b_n$.
\end{proposition}
\begin{proof}
Since for $i\ne k$ the $i$th equation in \refE{non_sys} does not explicitly depend on $a_k,b_k$,  by differentiation of those equations first in $a_k$, and then in $b_k$ we obtain two identical systems of linear equations on the partial derivatives of  $H_1,\ldots,H_n$:
\[
   \sum_{j=1}^n \frac{\partial F_i}{\partial H_j} \frac{\partial H_j}{\partial a_k}=0,\quad i=1,\ldots,n,\ i\ne k
\]
and
\[
   \sum_{j=1}^n \frac{\partial F_i}{\partial H_j} \frac{\partial H_j}{\partial b_k}=0,\quad i=1,\ldots,n,\ i\ne k
\]
By the assumption on the ranks we observe first that the vectors $(\frac{\partial H_1}{\partial a_k},\ldots,\frac{\partial H_n}{\partial a_k})$ and $(\frac{\partial H_1}{\partial b_k},\ldots,\frac{\partial H_n}{\partial b_k})$ are linearly dependent which gives the statement $1^\circ$, and second that both systems can be resolved with respect to $n-1$ unknowns corresponding to the nonzero determinant in~$2^\circ$.This way we obtain the relations \refE{rels2}
with the same function $A_{jkj_k}$ in both relations for every $j\ne j_k$ (by the similarity of the systems). We set $A_{j_kkj_k}\equiv 1$.
\end{proof}
\begin{proposition}\label{P:comm}
$H_1,\ldots,H_n$ commute with respect to any Poisson bracket of the above form on $\C^{2n}$ where $p_k=0$ if $k$ does not satisfy the assumption of \refP{rels1}. In particular, they commute with respect to the brackets corresponding to every pair of coordinates $(a_k,b_k)$ for $k$ satisfying the assumption.
\end{proposition}
\begin{proof}
It is obviously sufficient to prove the statement for $p_k\equiv 1$, $p_j\equiv 0$ ($j\ne k$) where $1\le k\le n$ is arbitrary.

By \refP{rels1}($1^\circ$) we have
\[
  \frac{\partial H_i}{\partial a_k}\frac{\partial H_j}{\partial b_k}=
  M_k\frac{\partial H_i}{\partial b_k}\frac{\partial H_j}{\partial b_k},
\]
and by \refP{rels1}($2^\circ$) we have
\[
  \frac{\partial H_i}{\partial a_k}\frac{\partial H_j}{\partial b_k}=
  A_{ikj_k}A_{jkj_k}\frac{\partial H_{j_k}}{\partial a_k}\frac{\partial H_{j_k}}{\partial b_k}.
\]
Both expressions on the right hand sides of those relations are symmetric in $i,j$, hence $\frac{\partial H_i}{\partial a_k}\frac{\partial H_j}{\partial b_k}-\frac{\partial H_j}{\partial a_k}\frac{\partial H_i}{\partial b_k}=0$.
\end{proof}
Coming to the map $\C^{2n}\to\C^n$ (resp., $\R^{2n}\to\R^n$) given by $(a,b)\to(H_1(a,b),\ldots,H_n(a,b))$ (where $a=(a_1,\ldots,a_n)$, $b=(b_1,\ldots,b_n)$), the pre-image of a point is given by Cartesian product of $n$ curves, each one given by the equation $F_i(H_1,\ldots,H_n,a_i,b_i)=0$ ($i=1,\ldots,n$) where $(H_1,\ldots,H_n)$ are regarded as  constants. The corresponding \emph{angle-type} coordinates are constructed in \cite{Skl} in the general set-up , and in \cite{DT} for the systems related to curves.

For a linear system of type \refE{bas_sys} with polynomial or analytical rank $n$ matrix $f(x,y)$, or for a system over the ring of germs of smooth functions, our \refP{comm} follows from the results of \cite{ER}. Under assumption that the function $f_{ij}$ is the same for all $i$ (for any $j$) the \refP{comm} follows also from \cite{BT1,BT2}. The Lagrange interpolation polynomial had been used in course of integrating certain classical non-linear systems \cite{BT1,BT2}, however without any explicit mentioning that its coefficients form an integrable system themselves. Similar ideas occur in \cite{KT} (see also references there). To the best of our knowledge the Hermit interpolation polynomial has been never mentioned with this relation. As well, \cite{BT2} is the only work we know where separation of variables is used for generating integrable systems. The present note develops this idea.

The author is grateful to O.Babelon, B.Dubrovin, S.Gusein-Zade, A.Khovanski, I.Krichever, S.Lando, A.Orlov, V.Rubtsov, D.Talalaev, V.Vasiliev for discussion.
\bibliographystyle{amsalpha}


\end{document}